\documentclass[journal,twoside,web]{ieeecolor}
\usepackage{generic}
\usepackage{cite}
\usepackage{multirow}%
\usepackage{amsmath,amssymb,amsfonts}
\usepackage{algorithmic}
\usepackage{graphicx}
\usepackage{algorithm,algorithmic}
\usepackage{hyperref}
\usepackage{textcomp}
\usepackage{bbm}
\usepackage{mathrsfs}%
\usepackage[percent]{overpic}

\newtheorem{theorem}{Theorem}

\newtheorem{remark}{Remark}%

\newcommand{\T}{\mathrm{T}}
\newcommand{\setS}{\mathcal{S}}

\def\BibTeX{{\rm B\kern-.05em{\sc i\kern-.025em b}\kern-.08em
    T\kern-.1667em\lower.7ex\hbox{E}\kern-.125emX}}
\begin{document}
\pagestyle{empty}
\thispagestyle{empty}
\title{Observing the state of networks with directed higher-order interactions}
\author{Roberto Rizzello, \IEEEmembership{Member, IEEE}, Davide Salzano, Stefano Boccaletti, and Pietro De Lellis, \IEEEmembership{Senior Member, IEEE}
\thanks{This study was carried out within the 2022FHHHPC ``The Structure, Dynamics and Control of Network Systems With Higher-Order Interactions'' project – funded by European Union – Next Generation EU  within the PRIN 2022 program (D.D. 104 - 02/02/2022 Ministero dell’Università e della Ricerca). This manuscript reflects only the authors’ views and opinions and the Ministry cannot be considered responsible for them.}
\thanks{Roberto Rizzello, Davide Salzano and Pietro De Lellis are with the Department of Electrical Engineering and Information Technology, University of Naples Federico II, 80125 Naples, Italy, (e-mail: roberto.rizzello@unina.it, davide.salzano@unina.it, pietro.delellis@unina.it).}
\thanks{Stefano Boccaletti is with the 
International Research Center of Complexity Sciences, Hangzhou International Innovation Institute,
Beihang University, 311115 Hangzhou, China, with the 
Sino-Europe Complexity Science Center, North University of China, 030051 Taiyuan, China, and with the CNR—Institute of Complex Systems, Via Madonna del Piano 10, I-50019 Sesto Fiorentino, Italy, (e-mail: stefano.boccaletti@gmail.com ).}}

\maketitle

\begin{abstract}
We consider the problem of reconstructing the state of a network of nonlinear dynamical systems in the presence of directed higher-order interactions. Grounded on analytical convergence results, we propose an algorithmic observer design procedure that simultaneously selects the nodes to be measured and the observer gains. We complement the theoretical analysis with an exhaustive numerical investigation campaign that showcases the performance and robustness of the designed observer. Finally, the algorithmic procedure is used to fully reconstruct the opinions of a group of agents.
\end{abstract}


\thispagestyle{empty}

\section{Introduction}
\label{sec:introduction}
\IEEEPARstart{M}{any} natural and engineered systems exhibit complex collective behavior as the result of the local interactions of the agents composing the system  \cite{thurner2018introduction,sayama2015introduction}. The formation patterns performed by swarms of drones \cite{kabore2021distributed,giusti2023distributed}, voting dynamics \cite{galam2004contrarian,sobkowicz2016quantitative}, or the synchronized functioning of power and smart grids \cite{dorfler2013synchronization} are all complex collective behaviors that can be described in terms of network dynamical systems. In the last decades, tools from dynamical systems and graph theory have been used to investigate the mechanisms fostering the onset of such behaviors \cite{boccaletti2006complex, newman2011structure}. Directed graphs have been widely used to describe the interactions between the dynamical units composing the networks, and analytical conditions determining the emergence of collective behaviors have been derived \cite{bullo2024lectures, strogatz2001exploring}.

Controlling the emergent dynamics of network systems is a pressing research challenge \cite{astrom2011control,liu2016control, siljak2011decentralized}, and several distributed algorithms based on state-feedback have been proposed to steer the network trajectories towards a desired one \cite{porfiri2008criteria, delellis2018partial, delellis2011pinning}. However, the state of the network nodes is not always measurable, either for the prohibitive cost of the required sensors \cite{yan2015spectrum}, or because the quantities to be measured are inaccessible \cite{liu2013observability}. This makes a direct application of state-feedback control unfeasible, and calls for the need of building state observers for the network. 

While observability in network systems has been widely studied, see e.g. \cite{montanari2020observability,iudice2019node,hashimoto2010controllability,subasi2014quantitative,bianchin2016observability,letellier2018symbolic,whalen2015observability} and references therein, few works have tackled the problem of designing an observer for the state of the entire network.
%
%
Observer design techniques for networks of linear dynamical systems have been studied in \cite{han2018simple, zhang2017distributed}. However, these works cannot be applied to nonlinear interconnected systems.
The case of Lipschitz individual dynamics has been investigated in
\cite{wan2017distributed}, whereas a dissipativity-based observer has been proposed in \cite{schmidt2014observer}.
Finally, an observer is employed in \cite{zhang2023intermediate} as an intermediate step towards solving a fault estimation problem.

Another underlying assumption on which these works rely on is that interactions are pairwise, that is, they involve only two nodes at a time. 
However, in recent years it has been highlighted that in several application fields \cite{boccaletti2023structure, battiston2021physics} coupling between agents cannot be restricted to be pairwise, needing the inclusion of higher-order interactions where three or more nodes are simultaneously involved. For example, the effect of multi-body interactions has been shown to be fundamental in social networks \cite{schawe2022higher,rizzello2024modeling}, in chemical reactions transforming multiple reagents in multiple products \cite{ajemni2017directed,mann2023ai}, or in multi-battery equalization protocols \cite{ouyang2024unified}.
To the best of our knowledge, there is no existing method to build an observer when the systems are coupled through (directed or undirected) hypergraphs.

The main contribution of this manuscript is the design of a state-observer for networks of nonlinear systems characterized by multibody directed higher-order interactions. Specifically, using a Lyapunov-based approach, we provide two alternative conditions for local convergence of the estimation error to zero. Based on these theoretical results, we propose a novel algorithm to design a state observer that, by iteratively exploring the network, simultaneously identifies a sufficient set of nodes whose outputs need to be measured and devises suitable observer gains. 
In addition, we show that the designed observer effectively copes with large initial estimation errors, and is robust with respect to uncertainties on the node individual parameters and measurement uncertainties.

The outline of the paper is as follows. Section \ref{sec:preliminaries} introduces the necessary notation, and provides background on hypergraphs, signed graphs and hyperdiffusive coupling protocols, whereas in Section \ref{sec:model} we describe the network model and formulate the observation problem. The conditions for the convergence of the observer dynamics to the network state are given in Section \ref{sec:conv}, and the design of the observer is illustrated in Section \ref{sec:obsv}. Section \ref{sec:validation} is then devoted to validate the proposed observer on numerical testbeds, and to test its robustness to parametric uncertainties. Finally, we apply the proposed observer in the context of opinion dynamics in Section \ref{sec:opinion}. Conclusions are drawn in Section \ref{sec:conclusion}.

\section{Mathematical Preliminaries}
\label{sec:preliminaries}

\color{black}
\subsection{Notation}
Given $n\in\mathbb{N}$, $\mathbbm{1}_n$ and $0_n$ are vectors of all ones and zeros in $\mathbb{R}^{n}$, respectively. The horizontal and vertical concatenations of the set of $p$ vectors $v_1,\ldots,v_p\in\mathbb{R}^n$ are denoted with $[v_1,\ldots,v_p]$ and $[v_1;\ldots;v_p]$, respectively. The identity matrix in $\mathbb R^{n\times n}$ is denoted $I_n$. Given a matrix $A \in \mathbb{R}^{n\times m}$, $A^\T$ is its transpose,  
$\underline{\sigma}(A)$ its smallest singular value, while $\|A\|$ and $\|A\|_F$ are its Euclidean and Frobenius norm, respectively. Given a square matrix, $A^{\mathrm{sym}}=(A+A^\T)/2$ is its symmetric part, and $\mathrm{spec}(A)$ is its spectrum. Moreover, $A\in\mathbb R^{n\times n}$, $A>0$ ($A\ge 0$) means that $A$ is positive (semi) definite. Similarly, $A<0$ ($A\le 0$) means that $A$ is negative (semi) definite. Moreover, given $A,B\in\mathbb R^n$, $A\le B$ means $A-B\le 0$.
Finally, given a vector field $\varphi(z):\mathbb R^m\to \mathbb R^n$, $D_z\varphi\in\mathbb R^{n\times m}$ denotes its Jacobian matrix with respect to $z$.

\subsection{Directed Hypergraphs and Signed Graphs}
A directed hypergraph $\mathscr{H}$ is defined as a pair of sets. The first, $\mathcal{V}$, contains the $N$ nodes composing the hypergraph, while the second, $\mathcal{E}$, its directed hyperedges. The generic directed hyperedge $\epsilon$ is in turn defined as a pair of ordered and disjoint subsets of $\mathcal{V}$. The first and second subsets contain the \emph{tails} and \emph{heads} of hyperedges $\epsilon$ that are denoted with $\mathcal{T}(\epsilon)$ and $\mathcal{H}(\epsilon)$, respectively \cite{de2022pinning}. The cardinalities of such subsets are denoted with ${|\mathcal{T}(\epsilon)|}$ and ${|\mathcal{H}(\epsilon)|}$, whereas the cardinality of $\epsilon$ is $|\epsilon|=|\mathcal{T}(\epsilon)|+|\mathcal{H}(\epsilon)|$. 
Given a subset of nodes $\mathcal V_{\mathrm{sub}}\subseteq \mathcal V$, we denote $\mathcal E_{\mathcal V_{\mathrm{sub}}}$ as the subset of hyperedges whose tails and heads belong to $\mathcal V_{\mathrm{sub}}$, that is, $\mathcal E_{\mathcal V_{\mathrm{sub}}}=\{\epsilon \in\mathcal E: \mathcal T(\epsilon)\cup\mathcal H(\epsilon)\subseteq \mathcal V_{\mathrm{sub}}\}$.
Given a subset of hyperedges $\mathcal E_{\mathrm{sub}}\subseteq \mathcal E$, $\mathcal T(\mathcal E_{\mathrm{sub}})$ is the union of the sets of tails of all hyperedges in $\mathcal E_{\mathrm{sub}}$, that is, $\mathcal T(\mathcal E_{\mathrm{sub}})=\bigcup_{\epsilon\in\mathcal E_{\mathrm{sub}}}\mathcal T(\epsilon)$. A sample hypergraph, with the standard representation of directed hyperedges is illustrated in Figure \ref{fig:sample_hypergraph}.

\begin{figure}[!ht]
\centering
\begin{overpic}[scale=0.25]{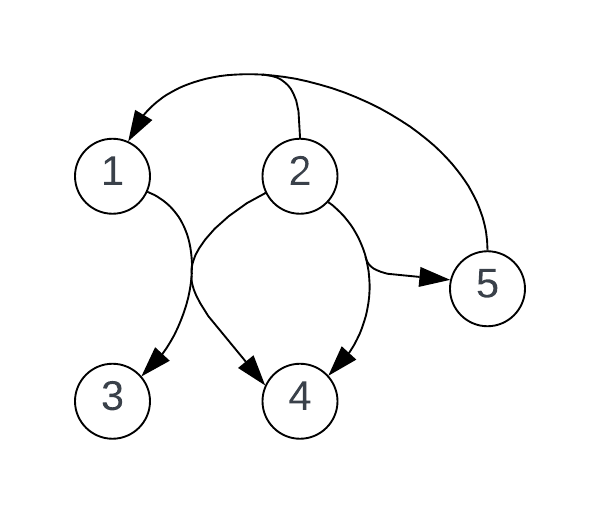}
\put (23,40) {$\epsilon_1$}
\put (53,40) {$\epsilon_2$}
\put (49,76) {$\epsilon_3$}
\end{overpic}
\caption{Sample directed hypergraphs with $\mathcal{V}=\{1,2,3,4,5\}$, $\mathcal{E}=\{\epsilon_1,\epsilon_2,\epsilon_3\}$. For example, the tail and head sets of hyperedge $\epsilon_1$ are $\mathcal{T}(\epsilon_1)=\{1,2\}$ and $\mathcal{H}(\epsilon_1)=\{3,4\}$, respectively.}
\vspace{-3mm}
\label{fig:sample_hypergraph}
\end{figure} 

Given $\mathcal{V}_1\subseteq\mathcal{V}$ and $\mathcal{V}_2\subseteq \mathcal{V}$, the hyperedge $\epsilon$ belonging to the set $\mathcal{E}^{\mathcal{V}_1,\mathcal{V}_2}\subseteq\mathcal{E}$ is such that $\mathcal{V}_1\subseteq \mathcal{T}(\epsilon) \land \mathcal{V}_2\subseteq \mathcal{H}(\epsilon)$. Similarly, all the hyperedges having node $j$ as a tail (head) and whose heads (tails) are elements of $\mathcal{\bar V}\subseteq\mathcal{V}$, belong to the set $\mathcal{E}^{j,\mathcal{\bar V}}$ ($\mathcal{E}^{\mathcal{\bar V},j}$). Finally, we denote with  $\mathcal{E}^{\star,j}\subseteq\mathcal{E}$ ($\mathcal{E}^{j,\star}\subseteq\mathcal{E}$) the subset containing all the hyperedges having node $j$ as a head (tail); we define the in-degree $d_j^{\mathrm{in}}$ and out-degree $d_j^{\mathrm{out}}$ of a node $\nu_j$ as the cardinality of the latter two sets, that is, $d_j^{\mathrm{in}}=|\mathcal{E}^{\star,j}|$ and $d_j^{\mathrm{out}}=|\mathcal{E}^{j,\star}|$.

A weighted signed graph $\mathcal G$ is defined by the triple $\{\mathcal V,\mathcal E, \mathcal Q\}$, where $\mathcal V$ is the set of nodes, $\mathcal E\subseteq \mathcal V\times \mathcal V$ is the set of edges, and the function $\mathcal Q: \mathcal V\times \mathcal V \rightarrow \mathbb R$ associates $0$ to each pair $(i,j)\in\mathcal V\times \mathcal V$ that is not in $\mathcal E$, and a non-zero weight to each edge in $\mathcal E$. Different from standard weighted digraphs, also negative weights can be associated to edges. Given a subset of nodes $\mathcal V_{\mathrm{sub}}\subseteq \mathcal V$, we denote $\mathcal G_{\mathcal V_{\mathrm{sub}}}$ the subgraph induced by $\mathcal V_{\mathrm{sub}}$.

\subsection{Hyperdiffusive Coupling Protocol}\label{sec:hyperdiff}
Next, let us associate to the $i$-th node in $\mathcal V$ a vector $x_i\in\mathbb R^n$ describing its state, for $i=1,\ldots,N$.
Following  \cite{della2023emergence}, given a hyperedge $\epsilon\in\mathcal E$,
we introduce a hyperdiffusive coupling protocol
\begin{equation} \label{eq:hyperdiffusive_coupling}
    g(x_\epsilon^{\tau} \alpha_\epsilon-x_\epsilon^{h} \beta_\epsilon),
\end{equation}
where $g:\mathbb{R}^n\to\mathbb{R}^n$ is the nonlinear coupling function. $g$ depends on the difference between the convex combinations of the states of the hyperedge's tails and heads, respectively. The combination coefficients $\alpha_\epsilon \in \mathbb{R}^{|\mathcal{T}(\epsilon)|}$ and $\beta_\epsilon\in \mathbb{R}^{|\mathcal{H}({\epsilon})|}$ are such that $\alpha_\epsilon^T\mathbbm{1}_{|\mathcal{T}(\epsilon)|}=\beta^T_\epsilon\mathbbm{1}_{|\mathcal{H}(\epsilon)|}=1$;
matrices $x_{\epsilon}^{\tau} \in \mathbb{R}^{n\times|\mathcal{T}(\epsilon)|}$ and $x_{\epsilon}^{h}\in \mathbb{R}^{n\times|\mathcal{H}(\epsilon)|}$, horizontally stack the states of the nodes that are tails and heads of the hyperedge $\epsilon$, respectively.

In \cite{della2023emergence}, it has been shown that, locally, the hyperdiffusive interactions over a directed hypergraph can be equivalently described by pairwise interactions over the associated signed graph. The association starts from the observation that the argument of the nonlinear function $g$ in \eqref{eq:hyperdiffusive_coupling} can be rewritten as
\begin{equation}
    \sum_{j\in\mathcal T(\epsilon)}(\tilde \alpha_\epsilon)_j(x_j-x_i)-\sum_{j\in\mathcal H(\epsilon)}(\tilde \beta_\epsilon)_j(x_j-x_i),
\end{equation}
where $\tilde \alpha_\epsilon\in\mathbb R^N$ ($\tilde\beta_\epsilon\in\mathbb R^N$) is a vector whose element $j$ is
0 if node $j$ is not a tail (head) of $\epsilon$, whereas, if $j$ is a tail
(head) of $\epsilon$, it is equal to the weight associated with that tail
(head). In simple terms, this implies that the associated signed graph will have positive edges from each tail to each head, and negative edges between the heads.

In the presence of hyperdiffusive coupling, we can then define the largest connected component of a directed hypergraph $\mathscr H=\{\mathcal V, \mathcal E\}$ as $\{\mathcal V_{\mathrm{sub}},\mathcal E_{\mathcal V_{\mathrm{sub}}}\}$, where $\mathcal V_{\mathrm{sub}}\subseteq \mathcal V$ is the set of nodes of the largest connected component of the associated signed graph.

\section{Network model and problem formulation}\label{sec:model}
%
%
%
\begin{figure}
    \centering
    \includegraphics[width=0.8\linewidth]{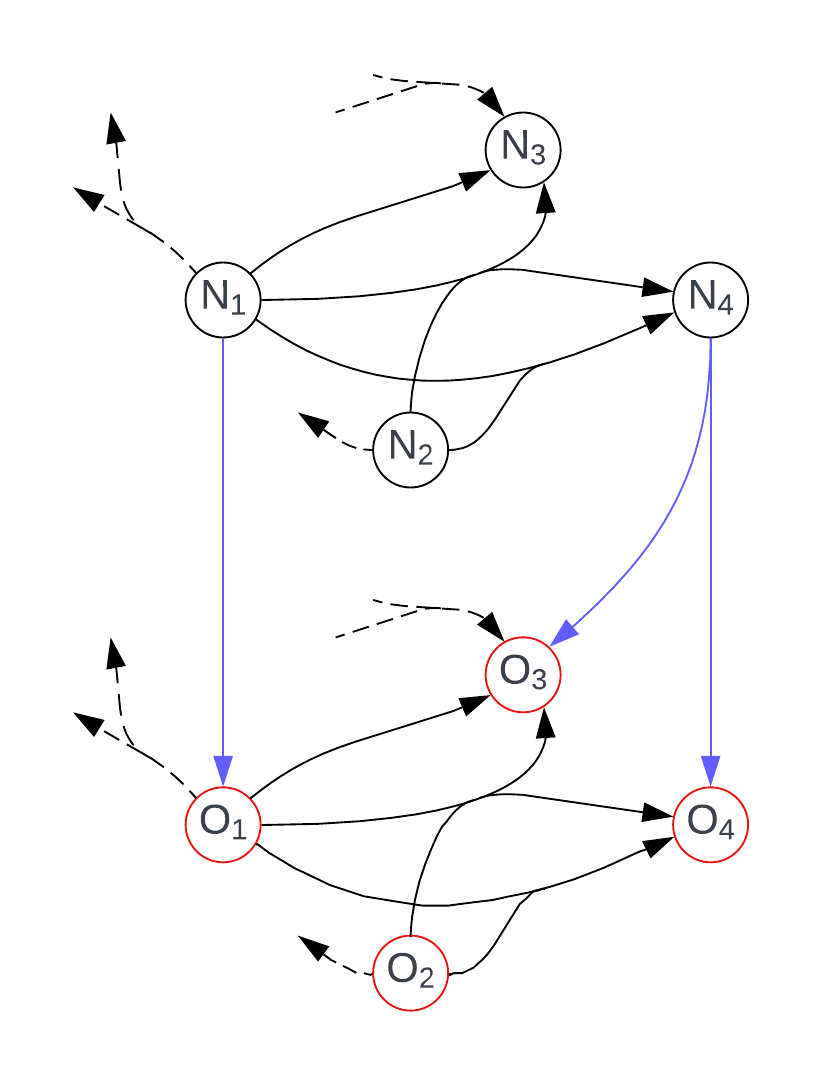}
    \caption{Schematic of the network state observer. The nodes of the network to be observed are represented in black, whereas the corresponding nodes of the observer are depicted in red. The nodes are coupled through a directed hypergraph, highlighting the presence of multi-body interactions. 
    The information flow from the observed network is represented by blue arrows: this means, for instance, that the output of node 4 is sent to nodes 3 and 4 of the observers.}
    \label{fig:master_slave_configuration}
\end{figure}
Let us consider a network of $N$ nodes coupled on a directed hypergraph $\mathscr H$, whose dynamics are given by 
\begin{equation}
\label{eq:network_dynamics}
\dot{x}_i= f(x_i)+\sum\limits_{\epsilon \in \mathcal{E}^{\star,i}} \sigma_{\epsilon} g(x_\epsilon^{\tau} \alpha_\epsilon-x_\epsilon^{h} \beta_\epsilon),\quad i=1,\ldots,N,
\end{equation}
%
where $g(x_\epsilon^{\tau} \alpha_\epsilon-x_\epsilon^{h} \beta_\epsilon)$ is the hyperdiffusive coupling protocol as defined in Section \ref{sec:hyperdiff}, and $\sigma_{\epsilon}$ is the coupling strength associated to hyperedge $\epsilon$.
%
%

Here, we consider the problem of observing the state of network \eqref{eq:network_dynamics}, when only the output of a subset $\mathcal O\subseteq \mathcal V$ of the nodes can be measured, that is,
\begin{equation}\label{eq:output}
    y_i=h(x_i),\qquad i\in\mathcal O,
\end{equation}
where $y_i\in\mathbb R^p$, $p\le n$, is the output of node $i$, and $h:\mathbb R^n\to\mathbb R^p$ is the output function.

Specifically, for a given output function $h$, we aim to design an observer and select a set $\mathcal O$ so that we can reconstruct the state of the entire network.

The dynamics we chose for our observer follows a prediction-correction paradigm, inspired to the classical Luemberger observer, and is given by
\begin{equation}
\label{eq:observer}
%
\dot{\hat x}_i=f(\hat x_i)+\sum_{j\in\mathcal O}L_{ij}( y_j- h (\hat x_j))+\sum\limits_{\epsilon \in \mathcal{E}^{\star,i}}\sigma_\epsilon g(\hat x_\epsilon^{\tau}\alpha_\epsilon-\hat x_\epsilon^{h} \beta_\epsilon),
\end{equation}
for $i=1,\dots,N$. Here, $L_{ij}:\mathbb R^{n\times p}$ is the correction matrix of the observer, and $h(\hat x_j)$ is the estimated output.

By defining $e_i=x_i-\hat x_i$ as the observation error of node $i$, we can then formally state our observation problem:

\textit{Observation problem.} Given network \eqref{eq:network_dynamics}, we aim to select the observed nodes $\mathcal O$ and the correction gain matrices $L_{ij}$, with $i\in\mathcal V$ and $j\in\mathcal O$, such that the observer dynamics \eqref{eq:observer} locally asymptotically converge to the network states (or, equivalently, the observation error locally asymptotically converges to zero), that is,
\begin{equation*}
\begin{aligned}
    &\exists \delta>0 : \Vert e(0) \Vert < \delta \implies \lim\limits_{t\to+\infty} e_i(t) = 0,\quad\\
    &\forall x_i(0)\in\mathcal W\subseteq \mathbb{R}^n, \, i=1,\ldots,N. 
\end{aligned}
\end{equation*}

\section{Convergence Analysis}\label{sec:conv}
We start by writing the error dynamics as 
\begin{align}\label{eq:error_non_lin_dynamics}
    \dot{e}_i&=f(x_i)-f(\hat x_i) - \sum_{j\in\mathcal O}L_{ij}(h(x_j)- h(\hat x_j))\nonumber\\&+\sum\limits_{\epsilon \in \mathcal{E}^{\star,i}} \sigma_\epsilon\big(g(x_\epsilon^{\tau}\alpha_\epsilon-x_\epsilon^{h} \beta_\epsilon)-g(\hat x_\epsilon^{\tau} \alpha_\epsilon-\hat x_\epsilon^{h} \beta_\epsilon)\big),
\end{align}
for $i=1,\ldots,N$.
Notice that $e_i(t) = 0$ for all $i=1,\ldots,N$ is an equilibrium point for \eqref{eq:error_non_lin_dynamics}, and we can then study its local stability by linearizing its dynamics \eqref{eq:error_non_lin_dynamics} around it, thus obtaining 
\begin{align}\label{eq:final_form}
    \dot {\tilde e}_i&=D_xf(\hat x_i) \tilde e_i-\sum_{j\in\mathcal O}L_{ij}D_xh(\hat x_j)\tilde e_j\nonumber \\&
    +\sum\limits_{\epsilon \in \mathcal{E}^{\star,i}}\sigma_\epsilon D_xg(\hat x_\epsilon^{\tau} \alpha_\epsilon-\hat x_\epsilon^{h} \beta_\epsilon)(\tilde e_\epsilon^{\tau} \alpha_\epsilon-\tilde e_\epsilon^{h} \beta_\epsilon)
\end{align}
for all $i=1,\ldots,N$,
where $\tilde e_i$ is the linearized observation error of node $i$ and matrices $\tilde e_\epsilon^\tau$ and $\tilde e_\epsilon^h$ horizontally stack the linearized errors of the nodes that are tails and heads of the hyperedge $\epsilon$, respectively.

Let us consider a subset $\mathcal S=\{s_1,\ldots,s_M\}\subseteq\mathcal  V$ of nodes.

Next, we define matrix $A_{\mathcal S}(t)\in\mathbb R^{nM\times nM}$ as
\begin{align}\label{eq:A_S}
    A_{\mathcal{S}}(t)=\begin{bmatrix}
        A_{11} & \ldots &  A_{1M} \\
        \vdots & \ddots & \vdots\\
        A_{M1} & \ldots &  A_{MM}
    \end{bmatrix},
\end{align}
with, for all $i=1,\ldots,M$,
\begin{align}\label{eq:Phi_term}
        A_{ii}&=D_xf(\hat x_{s_i})-\mathcal I_{s_i}(\mathcal O)L_{s_i s_i}D_xh(\hat x_{s_i})\nonumber\\
        &-\sum\limits_{\epsilon\in\mathcal{E}^{\star,s_i}}\sigma_{\epsilon}(\beta_\epsilon)_{s_i}D_xg(\hat x_\epsilon^{\tau}\alpha_\epsilon-\hat x_\epsilon^{h}\beta_\epsilon),
    \end{align}
\begin{align}\label{eq:Gamma_term}
    A_{ij}&=-\mathcal I_{s_j}(\mathcal O)L_{s_i s_j}D_xh(\hat x_{s_j})\\
    &+\sum\limits_{\epsilon\in\mathcal{E_S}^{s_j,s_i}}\sigma_{\epsilon}(\alpha_\epsilon)_{s_j}D_xg(\hat x_\epsilon^{\tau}\alpha_\epsilon-\hat x_\epsilon^{h}\beta_\epsilon)\nonumber\\
    &-\sum\limits_{\epsilon\in\mathcal{E}^{\star,(s_i,s_j)}}\sigma_{\epsilon}(\beta_\epsilon)_{s_j}D_xg(\hat x_\epsilon^{\tau}\alpha_\epsilon-\hat x_\epsilon^{h}\beta_\epsilon),\nonumber
\end{align}
for $j=1,\ldots,M, j\ne i$. In \eqref{eq:Phi_term} and \eqref{eq:Gamma_term}, $\mathcal I_{z}(\mathcal O)$ is the indicator function, which is 1 if $z\in\mathcal O$, and 0 otherwise.

Next, we define $b_\mathcal{S}(t)=[b_{s_1}; \ldots; b_{s_M}]\in \mathbb{R}^{nM}$, with 

\begin{equation}\label{eq:b_term}
        b_{s_i}=\sum\limits_{\epsilon\in\mathcal{E}^{\star,s_{i}}}\sigma_\epsilon D_xg(\hat x_\epsilon\alpha_\epsilon-\hat x_\epsilon\beta_\epsilon)\sum_{\tau(\epsilon,j)\notin \mathcal S}(\alpha_\epsilon)_j\tilde e_
        {\tau(\epsilon,j)},
\end{equation}
where $\tau(\epsilon,j)$ is the function that associates to the $j$-th tail of $\epsilon$ its label in $\mathcal V$. Finally, let us define $\hat x=[\hat x_1;\dots;\hat x_N]$. 
    
We can now state the following theorems:

\begin{theorem}\label{thm:slowly_varying}
Given a subset of nodes $\mathcal S=\{s_1,\ldots,s_M\}\subseteq\mathcal  V$. If
    \begin{enumerate}
        \item [H1)] $A_\setS(t)$ and $D_x g(t)$ are bounded for all $t$, and $A_\setS(t)$ has a uniform Hurwitz margin;
        \item [H2)] there exists a positive scalar $\varepsilon$ such that $$\|\dot A_\setS(t)\|<\frac{\underline{\sigma}(Q_\setS)\underline{\sigma}(A_\setS(t)\oplus A_\setS(t))^2}{2\|Q_\setS\|_{\mathrm{F}}}(1-\varepsilon)$$ for all $\hat x(t)$, with $Q_\setS \in \mathbb{R}^{nM\times nM}$ being a positive definite symmetric matrix;
        \item [H3)] $ e_j(t)$ locally asymptotically converges to zero for all $j\in\mathcal{T}(\mathcal{E}^{\star,S}\setminus\mathcal{E}_\setS)$;
        
    \end{enumerate}
    
    then the error $e_{s_j}(t), j=1,\ldots,M,$ locally asymptotically converges to zero.
\end{theorem}

\begin{proof}
    We study the stability properties of the linearized error dynamics associated to the nodes belonging to $\mathcal{S}$, which are described by
    \begin{align}\label{eq:es_dynamics}
        \dot {\tilde e}_{\mathcal{S}}=A_{\mathcal{S}}(t)\tilde e_\mathcal{S} +b_{\mathcal{S}}(t),
    \end{align}
    where $\tilde e_\setS=[\tilde e_{s_1},\ldots,\tilde e_{s_M}]$.
     From H1 and \cite{amato2002new}, there exists a time-varying matrix $P(t)$ symmetric and positive definite such that 
     \begin{equation}\label{eq:Lyap_eq}
A_\mathcal{S}(t)^\T P_\mathcal{S}(t)+P_{\mathcal{S}}(t)A_{\mathcal{S}}(t)=-Q_\setS.
    \end{equation}
We can then introduce the following Lyapunov function candidate: 
    \begin{equation}
        V = \frac{1}{2}\tilde e_\setS^\T P_\setS(t)\tilde e_\setS,
    \end{equation}
    where $P_\setS(t)>0$ for all $t$. The time-derivative of $V$ then reads
    \begin{equation}\label{eq:Lyap_derivative}
        \dot V=\Big[(\tilde e_\setS^\T A_\setS(t)^\T+b_\setS^\T)+\tilde e_\setS^\T\dot P_\setS\Big]\tilde e_\setS+\tilde e_\setS^\T P_\setS(A_\setS(t)^\T \tilde e_\setS+b_\setS(t)).
    \end{equation}
    From \eqref{eq:Lyap_eq}, and noting that $\dot P_\setS\le \|\dot P_\setS\|I_{nM}$ and $-Q_\setS\le -\underline{\sigma}(Q_\setS)I_{nM}$, we obtain
    \begin{align}\label{eq:norm_boundingp}
        \dot V(\tilde e_\setS)&=-\tilde e_\setS^\T(Q_\setS-\dot P_\setS)\tilde e_\setS+2\tilde e_\setS^\T P_\setS b_\setS(t) \nonumber\\& 
        \le-\Bigg(1-\frac{\|\dot P_\setS\|}{\underline{\sigma}(Q_\setS)}\Bigg)\underline{\sigma}(Q_\setS)\|\tilde e_\setS\|^2+2\tilde e_\setS^\T P_\setS b_\setS(t)
    \end{align}
From \cite[eq. (10)]{amato2002new}, we know that
    \begin{equation}\label{eq:bound_garofalo}
        \|\dot P_\setS\|\leq\frac{2\|\dot A_\setS(t)\|\|Q_\setS\|_{\mathrm{F}}}{\underline{\sigma}^2(A_\setS(t)\oplus A_\setS(t))}.
    \end{equation}
    Consequently, we have
    \begin{align}\label{eq:norm_bounding}
        &\dot V\leq-\Bigg(1-\frac{2\|\dot A_\setS(t)\|\|Q_\setS\|_{\mathrm{F}}}{\underline{\sigma}(A_\setS(t)\oplus A_\setS(t))^2\underline{\sigma}(Q_\setS)}\Bigg)\underline{\sigma}(Q_\setS)\|\tilde e_\setS\|^2\nonumber\\
        &+2\tilde e_\setS^\T P_\setS b_\setS(t)
        =-(1-\varsigma_\setS(t))\underline{\sigma}(Q_\setS)\|\tilde e_\setS\|^2+2\tilde e_\setS^\T P_\setS b_\setS(t),
    \end{align}
    where $0<\varsigma_\setS(t)=\frac{2\|\dot A_\setS(t)\|\|Q_\setS\|_{\mathrm{F}}}{\underline{\sigma}(A_\setS(t)\oplus A_\setS(t))^2\underline{\sigma}(Q_\setS)}< 1-\varepsilon$ from H2.
    Additionally, given that 
    \begin{equation*}
    -||\tilde e_\setS||^2\leq
    -\frac{\|P_\setS^{\frac{1}{2}}\tilde e_\setS\|^2}{\|P_\setS^{\frac{1}{2}}\|^2} =
    - \frac{2V}{\|P_\setS^{\frac{1}{2}}\|^2},
    \end{equation*}
    we can write 
    \begin{equation}
    \dot V\leq-(1-\varsigma_\setS(t))\underline{\sigma}(Q_\setS)\frac{2V}{\|P_\setS^{\frac{1}{2}}\|^2}+2\tilde e_\setS^\T P_\setS b_\setS(t)
    \end{equation}
    Since 
    \begin{multline*}
        2\tilde e_\setS^\T P_\setS b_\setS(t)\leq 2\|P_\setS^\frac{1}{2}\|\|P_\setS^\frac{1}{2}\tilde e_\setS\|_2\|b_\setS(t)\|\\ = 2\sqrt{2}\sqrt{V}\|P_\setS^\frac{1}{2}\|\|b_\setS(t)\|,
    \end{multline*}%
    we can write
    \begin{equation}
        \dot V\leq-(1-\varsigma_\setS(t))\underline{\sigma}(Q_\setS)\frac{2V}{\|P_\setS^\frac{1}{2}\|^2}+2\sqrt{2}\sqrt{V}\|P_\setS^\frac{1}{2}\|\|b_\setS(t)\| .
    \end{equation}

    Finally, from H1, there exists a finite scalar $\xi$ such that $\|P_\setS^{\frac{1}{2}}\|\leq \xi$. Therefore,
    \begin{equation} \label{eq:V_comparison}
        \dot V\leq-d_1(t) V+d_2\sqrt{V}\|b_\setS(t)\|,
    \end{equation}
where $d_1(t)=\big(2(1-\varsigma_\mathcal S)\underline{\sigma}(Q_{\mathcal S}) \big)/\xi^2$, and $d_2=2\sqrt{2}\xi$.

    Now, taking $z = \sqrt{V}$, we have
    \begin{equation}
        \dot z = \frac{\dot V}{2\sqrt{V}} \leq -\frac{d_1(t)}{2} z+\frac{d_2}{2}\|b_\setS(t)\|.
    \end{equation}

    Recalling that $\varsigma_\setS(t)<1-\varepsilon$, there exists a positive scalar $c$ such that $-d_1 < -2c $. Hence,
    \begin{equation} \label{eq:Z_dyn}
        \dot z \leq -cz+\frac{d_2}{2}\|b_\setS(t)\|
    \end{equation}
%
    
    %
%
From H1 and H3, $\lim\limits_{t\to+\infty}(\|b_\setS(t)\|)=0$, hence the thesis follows.
    

\end{proof}
\begin{theorem}\label{thm:lmi}
    Given a subset of nodes $\mathcal S=\{s_1,\ldots,s_M\}\subseteq\mathcal  V$. If
    \begin{enumerate}
        \item [H1)] $A_\setS^{\mathrm{sym}}(t)$ has a uniform Hurwitz margin for all $t$;
        \item [H2)] $D_x g(t)$ is bounded for all $t$;
        \item [H3)] $e_j(t)$ locally asymptotically converges to zero  for all $j\in\mathcal{T}(\mathcal{E}^{\star,S}\setminus\mathcal{E}_\setS)$;
        
    \end{enumerate}
    
    then the error $e_{s_j}(t), j=1,\ldots,M,$ locally asymptotically converges to zero.
\end{theorem}
\begin{proof}
    Let us consider the following Lyapunov function candidate:
    \begin{equation}
        V=\frac{1}{2}\tilde e_\setS^\T \tilde e_\setS.
    \end{equation}
From \eqref{eq:es_dynamics}, the time-derivative of $V$ reads
\begin{equation}
    \dot V=\tilde e_\setS^\T A_\setS(t) \tilde e_\setS+\tilde e_\setS^\T b_\setS=\tilde e_\setS^\T A_\setS^\mathrm{sym}(t) \tilde e_\setS+\tilde e_\setS^\T b_\setS(t)
\end{equation}
From the maximum principle for symmetric matrices, $\tilde e_\setS^\T A_\setS^\mathrm{sym}(t) \tilde e_\setS\le\lambda_{\max}(A_\setS^{\mathrm{sym}}(t))\tilde e_\setS^\T \tilde e_\setS$. Moreover, $\|\tilde e_\setS^\T b_\setS(t)\|\le\|\tilde e_\setS\| \| b_\setS(t)\|$. We then have
\begin{equation}
    \dot V\le-d_1(t) V+d_2\sqrt{V}\|b_\setS(t)\|,
\end{equation}
where $d_1(t)=-2\lambda_{\max}(A_{\mathcal S}^{\mathrm{sym}}(t))$, and $d_2=\sqrt{2}$. Note that, from H1, there exists a positive scalar $\phi$ such that $d_1(t)\ge \phi$, and that from H2 and H3 $\lim\limits_{t\to+\infty}(\|b_\setS(t)\|)=0$.  Similar steps to those reported after equation \eqref{eq:V_comparison} in the proof of Theorem \ref{thm:slowly_varying} then yield the thesis.
\end{proof}

\begin{remark}
    To ensure that it is possible to reconstruct the state of the network starting from an arbitrary initial condition $x_i(0) \in \mathcal W, \, i = 1,\ldots,N$, the hypotheses of theorems \ref{thm:slowly_varying} and \ref{thm:lmi} must be verified for all the trajectories rooted in $\mathcal W^N$. 
    
\end{remark}

\section{Design of the observer}\label{sec:obsv}

Here, we show how Theorems \ref{thm:slowly_varying} and \ref{thm:lmi} can be used to select the nodes that need to be measured to design the state observer that reconstructs the full state of the entire network. The critical network nodes that enable a complete observation are found through an algorithm that recursively explores the hypergraph topology.

As discussed in Section \ref{sec:hyperdiff}, the hyperdiffusive interactions over a directed hypergraph can be locally described by pairwise interactions over the associated signed graph. This implies that, locally, the dependencies between the node dynamics can be described by the signed graph $\mathcal G(\mathscr H)$ associated to $\mathscr H$. Accordingly, our algorithm will directly explore $\mathcal G(\mathscr H)$, leveraging well defined concepts on (signed) graphs, such as the graph condensation and source nodes \cite{bullo2024lectures}.

Denoting $\mathcal V_o\subseteq \mathcal V$ the set of nodes that we can successfully observe, that is, such that 
\[
\lim_{t\rightarrow+\infty}\tilde e_i(t)=0, \quad i\in\mathcal V_o,
\]
our algorithm aims to iteratively add nodes to $\mathcal O$ until $\mathcal V_o=\mathcal V$. Algorithm \ref{alg:main}, whose pseudo-code is reported for clarity below, is initialized by setting $\mathcal V_o=\emptyset$. Defining $\mathcal V_r=\mathcal V\setminus\mathcal V_o$, we construct the subgraph $\mathcal G_{\mathcal V_r}$ to then build the condensation graph $\mathcal C(\mathcal G_{\mathcal V_r})$. We identify the root strongly connected components (RSCC) of $\mathcal G_{\mathcal V_r}$, which correspond to the source nodes in $\mathcal C$. By construction, 
taking $\mathcal S$ as the nodes in the RSCCs, the set $\mathcal T(\mathcal E^{\star,{\mathcal S}}\setminus \mathcal E_{\mathcal S})\subseteq \mathcal V_o$, and therefore hypothesis H3 of Theorems \ref{thm:slowly_varying} and \ref{thm:lmi} is fulfilled. For each RSCC, we check whether the remaining hypothesis of Theorems \ref{thm:slowly_varying} and \ref{thm:lmi} are fulfilled for all $x_i(t)$ such that $x_i(0)\in\mathcal W$, $i\in\mathcal S$. If the hypotheses of any of the two theorems are fulfilled, we add the nodes in $\mathcal S$ to $\mathcal V_o$, otherwise we add a node $i\in\mathcal S$ to $\mathcal O$ following heuristic Algorithm \ref{alg:selection}, and repeat the procedure until the hypotheses of any of the two theorems are fulfilled, or $\mathcal S\subseteq\mathcal O$. If the hypotheses are not fulfilled even when $\mathcal S=\mathcal O$, the algorithm fails to guarantee observability of the full network state. The construction of $\mathcal{G}_{\mathcal{V}_{\mathrm{r}}}$ is repeated until $\mathcal V_r\neq \emptyset$.

Notice that fulfilling the hypothesis of Theorems \ref{thm:slowly_varying} and \ref{thm:lmi} depends on the selection of matrices $L_{ij}$ for  $i\in\mathcal S$, $j\in\mathcal O$. For Theorem \ref{thm:slowly_varying}, we choose $L_{ij}$ such that $\mathrm{spec}(A_\setS(t))=\Lambda^\star$, where $\Lambda^\star$ is selected so that the dynamics associated to the dominant eigenvalue is faster than the network dynamics. This ensures that H1 is fulfilled. Instead, for Theorem \ref{thm:lmi}, if possible, $L_{ij}$ is chosen as a feasible solution of the linear matrix inequality $A_\setS^{\mathrm{sym}} < \varepsilon I_{nM}$, where $\varepsilon$ is chosen so that the dynamics of the observer is faster than the dynamics of the observed network system.

\begin{algorithm}
\caption{Observer design algorithm}
    \label{alg:main}
\begin{algorithmic}
\STATE 
\STATE  \textbf{Input} $\mathscr{H}$, $D_x f$, $D_x g$, $D_x h$
\STATE \textbf{Output} $\mathcal O$

\STATE
\STATE \textbf{Initialization} 
\STATE $\mathcal V_o= \emptyset$; 
\STATE $\mathcal{V}_r= \mathcal{V}$;
\STATE Compute the signed graph $\mathcal G(\mathscr H)$ associated to $\mathscr H$;

\STATE
\STATE \textbf{Main loop}
\STATE
\textbf{while} $\mathcal V_r\ne \emptyset$:
\STATE \hspace{0.5cm} Extract the subgraph $\mathcal G_{\mathcal V_r}$; 
\STATE \hspace{0.5cm} Compute the condensation $\mathcal C(\mathcal G_{\mathcal V_r})$;
\STATE \hspace{0.5cm} Using $\mathcal C(\mathcal G_{\mathcal V_r})$, find the RSCCs $\mathcal S$ of $G_{\mathcal V_r}$;
\STATE \hspace{0.5cm} \textbf{for} all $\mathcal S$:
\STATE \hspace{0.5cm} \textbf{do}:
\STATE \hspace{1cm} Design the observer to fulfill the hypotheses of
\STATE \hspace{1cm} Theorems \ref{thm:slowly_varying} and \ref{thm:lmi} using $D_xf$ and $D_xg$;
\STATE \hspace{1cm} \textbf{if} the hypotheses of any of the two theorems are
\STATE \hspace{1.34cm} fulfilled:
\STATE \hspace{1.5cm} Add $\mathcal S$ to $\mathcal V_o$;
\STATE \hspace{1cm} \textbf{else}:
\STATE \hspace{1.5cm} Add a node in $\mathcal S$ to $\mathcal O$ according to
\STATE \hspace{1.5cm} Algorithm \ref{alg:selection};
\STATE \hspace{0.5cm} \textbf{while} the hypotheses of Theorems \ref{thm:slowly_varying} or \ref{thm:lmi} are not 
\STATE \hspace{1.52cm}fulfilled;   
\end{algorithmic}
\end{algorithm}

\begin{algorithm}
\caption{Selection of the node in $\mathcal S$}
    \label{alg:selection}
\begin{algorithmic}
\STATE 
\STATE  \textbf{Input} $\mathcal{S}$, $\mathscr H$, $\mathcal G_{\mathcal S}$
\STATE \textbf{Output} $j\in\mathcal S$

\STATE \textbf{for} all $i\in\mathcal S$
\STATE \hspace{0.5cm} Compute $d_i^{\mathrm{in}}$, $d_i^{\mathrm{out}}$ considering only the hyperedges 
\STATE \hspace{0.5cm} in $\mathcal E_{\mathcal S}$;
\STATE \hspace{0.5cm} Find $\mathcal R={\arg\max}_{i} d_i^{\mathrm{out}}$;
\STATE \hspace{0.5cm} \textbf{if} $|\mathcal R|=1$:
\STATE \hspace{1cm} $j=\mathcal R$;
\STATE \hspace{0.5cm} \textbf{else}:
\STATE \hspace{1cm} $j$ is the first element of ${\arg \max}_i d_i^{\mathrm{in}}$   
\end{algorithmic}
\end{algorithm}

\begin{remark}
When the output function $h$ is invertible, we use a simplified version of the observer design algorithm. Specifically, we can exactly reconstruct the state of a measured node, say $i$, from the initial time as $x_i(0)=h^{-1}(y_i(0))$. Therefore, in Algorithm \ref{alg:main} we can directly add node $i$ to $\mathcal V_o$ without the need of checking the hypotheses of Theorems \ref{thm:slowly_varying} or \ref{thm:lmi}.
\end{remark}

\begin{remark}
During the execution of the main loop of Algorithm \ref{alg:main}, the hypotheses of Theorems \ref{thm:slowly_varying} and \ref{thm:lmi} are checked on a representative set of trajectories rooted in $\mathcal{W}$ to reduce the computational effort.
\end{remark}

\section{Observing under parametric and measurement uncertainties}

Here, we consider the case where the individual dynamics have parametric mismatches. Namely, we consider that the vector field of node $i$ is $f(x_i,\mu_i)$, with $\mu_i=\hat \mu+\delta\mu_i$, where $\hat \mu$ is the nominal value of the parameters, and $\delta\mu_i$ is the deviation of the parameters of node $i$ from $\hat \mu$. We assume this deviation to be bounded, that is, $\|\delta\mu_i\|\le \bar \mu$. Network dynamics \eqref{eq:network_dynamics} then becomes
\begin{equation}
\label{eq:network_dynamics_het}
\dot{x}_i= f(x_i,\mu_i)+\sum\limits_{\epsilon \in \mathcal{E}^{\star,i}} \sigma_{\epsilon} g(x_\epsilon^{\tau} \alpha_\epsilon-x_\epsilon^{h} \beta_\epsilon),\quad i=1,\ldots,N,
\end{equation}
Additionally, we consider the presence of bounded measurement noise, whereby equation \eqref{eq:output} becomes
\begin{equation}\label{eq:output_noise}
    y_i=h(x_i)+\nu_i(t),\qquad i\in\mathcal O,
\end{equation}
where $\|\nu_i(t)\|\le \bar\nu$.

In this setting, we adjust the observer structure by setting its parameter vector to $\hat \mu$, that is,
\begin{equation}
\label{eq:observer_uncertain}
\begin{aligned}
\dot{\hat x}_i&=f(\hat x_i,\hat\mu)+\sum_{j\in\mathcal O}L_{ij}( y_j- h (\hat x_j))\\
&\hspace{12pt}+\sum\limits_{\epsilon \in \mathcal{E}^{\star,i}}\sigma_\epsilon g(\hat x_\epsilon^{\tau}\alpha_\epsilon-\hat x_\epsilon^{h} \beta_\epsilon).
\end{aligned}
\end{equation}
Following the same line of argument as in Section \ref{sec:obsv}, we can write the linearized error dynamics as
\begin{align}\label{eq:final_form_noise}
    \dot {\tilde e}_i&=D_xf(\hat x_i,\hat\mu) \tilde e_i+D_\mu f(\hat x_i,\hat\mu)\delta\mu_i-\sum_{j\in\mathcal O}L_{ij}D_xh(\hat x_j)\tilde e_j\nonumber \\
    &-\sum_{j\in\mathcal O}L_{ij}\nu_j
    +\sum\limits_{\epsilon \in \mathcal{E}^{\star,i}}\sigma_\epsilon D_xg(\hat x_\epsilon^{\tau} \alpha_\epsilon-\hat x_\epsilon^{h} \beta_\epsilon)(\tilde e_\epsilon^{\tau} \alpha_\epsilon-\tilde e_\epsilon^{h} \beta_\epsilon).
\end{align}
Given a subset of nodes $\mathcal S=\{s_1,\ldots,s_M\}\subseteq\mathcal  V$, the error dynamics of the subset can be written as
\begin{equation}
    \dot {\tilde e}_{\mathcal S}=A_\mathcal{S}(t)\tilde e_\mathcal S + b_\mathcal S+b^\mu_\mathcal S+b^\nu_\mathcal S,
\end{equation}
where $b^\mu_\mathcal S=[D_\mu f(\hat x_{s_1},\hat \mu)\delta\mu_{s_1};\ldots;D_\mu f(\hat x_{s_M},\hat \mu)\delta\mu_{s_M}]$, and $b^\nu_\mathcal S=[\sum_{j\in\mathcal O}L_{s_1 j}\nu_j;\ldots;\sum_{j\in\mathcal O}L_{s_M j}\nu_j]$.

Furthermore, we define
         \begin{equation}
        L_\setS=\begin{bmatrix}
            (L_\setS)_{11} & \ldots & (L_\setS)_{1M}\\
            \vdots & \ddots & \vdots \\
            (L_\setS)_{M1} & \ldots & (L_\setS)_{MM},
        \end{bmatrix}.
    \end{equation}   
    with $(L_\setS)_{ij}\in\mathbb R^{n\times p}$ being
    \begin{equation}
        (L_{\mathcal S})_{ij}=\left\{
        \begin{aligned}
            &L_{s_i s_j},\quad  s_j\in\mathcal O,\\
            &0_{n\times p}, \hspace{11pt} \text{otherwise.}
        \end{aligned}\right.
    \end{equation}

Next, we aim to guarantee a locally bounded observation error, that is, there exists a $b\ge 0$ such that
\begin{enumerate}
    \item for each $\xi> b$, there exist scalars $\Delta_1^x(\xi)$ and $\Delta_1^\mu(\xi)$ such that
    \begin{equation}
        \|e(t)\|\le \xi, \quad \forall t\ge 0, 
    \end{equation}
    for all $\|e(0)\|\le \Delta_1^x(\xi)$, $\|\delta \mu\|\le \Delta_1^\mu(\xi)$; and
    \item there exist scalars $\Delta_2^x(\xi)$ and $\Delta_2^\mu(\xi)$ such that
    \begin{equation}
    \limsup_{t\rightarrow+\infty}\|e(t)\|\le b, 
    \end{equation}
    for all $\|e(0)\|\le \Delta_2^x(\xi)$, $\|\delta \mu\|\le \Delta_2^\mu(\xi)$.
\end{enumerate}

We can now characterize the robustness of our observer using the following theorem:
\begin{theorem}\label{thm:rob}
    Given a subset of nodes $\mathcal S=\{s_1,\ldots,s_M\}\subseteq\mathcal  V$. If
    \begin{enumerate}
        \item [H1)] $A_\setS^{\mathrm{sym}}(t)$ has a uniform Hurwitz margin $\mathfrak h$ for all $t$;
        \item [H2)] $\|D_x g(t)\|\le\mathfrak g$, and $\|D_\mu f(t)\|\le \mathfrak f$ for all $t$;
        \item [H3)] $\limsup\limits_{t\to+\infty}\|\tilde e_j(t)\|\le B_j,$ for all $j\in\mathcal{T}(\mathcal{E}^{\star,S}\setminus\mathcal{E}_\setS)$;
        \item [H4)] $\|\delta\mu_i\|\le \bar \mu$ for all $i\in\mathcal V$, and $\|\nu_i(t)\|\le \bar \nu$ for all $i\in\mathcal O$, $t\ge 0$;
        \item [H5)] $L_{ij}=0$ for all $i\in\mathcal S$, $j\notin \mathcal S$.
    \end{enumerate}
    Then $e_{\mathcal S}(t)$ is locally bounded, and
    $$
            \limsup\limits_{t\to+\infty}\|\tilde e_{\mathcal S}(t)\|\le B,
    $$
where
\begin{equation}
    B=\frac{\sqrt{2}\sum_{j=1}^M \bar b_{s_j}+\sqrt{2 M}\bar\mu\mathfrak f+\sqrt{2 M}\bar\nu\mathfrak l}
    {2\mathfrak h},
\end{equation}
with $\mathfrak l=\limsup_{t\to+\infty}\|L_\setS(t)\|$, and
\[
\bar b_{s_j}=\mathfrak g \sum\limits_{\epsilon\in\mathcal{E}^{*,s_{i}}}\sigma_\epsilon \sum_{\tau(\epsilon,j)\notin \mathcal S} B_j.
\]
    \end{theorem}
\begin{proof}
    Consider the Lyapunov function candidate:
    \begin{equation}
        V = \frac{1}{2}\tilde e_\setS^\T \tilde e_\setS,
    \end{equation}
    By computing its time derivative, one obtains
    \begin{equation}\label{eq:dotv_th3}
    \begin{aligned}
        \dot V&=\tilde e_\setS^\T A_\setS(t) \tilde e_\setS+\tilde e_\setS^\T b_\setS+\tilde e_\setS^\T b_\setS^\mu+\tilde e_\setS^\T b_\setS^\nu\\&=\tilde e_\setS^\T A_\setS(t)^{\mathrm{sym}} \tilde e_\setS+\tilde e_\setS^\T b_\setS+\tilde e_\setS^\T b_\setS^\mu+\tilde e_\setS^\T b_\setS^\nu
    \end{aligned}
    \end{equation}
    From the maximum principle for symmetric matrices, we know that 
    \begin{equation}\label{eq:max_pr}
    \tilde e_\setS^\T A_\setS^\mathrm{sym}(t) \tilde e_\setS\le\lambda_{\max}(A_\setS^{\mathrm{sym}}(t))\tilde e_\setS^\T \tilde e_\setS.
    \end{equation} 
    Additionally, 
    \begin{align}
    \|\tilde e_\setS^\T b_\setS\|&\le\|\tilde e_\setS\| \| b_\setS\|,\\
        \|\tilde e_\setS^\T b_\setS^\mu\|\le\|\tilde e_\setS\| \| b_\setS^\mu\|&\leq \|\tilde e_\setS\|\|D_\mu f(t)\|\bar\mu\sqrt{M},\\
        \|\tilde e_\setS^\T b_\setS^\nu\|\le\|\tilde e_\setS\| \| b_\setS^\nu\|&\leq \|\tilde e_\setS\| \|L_\setS(t)\|\bar\nu\sqrt{M},\label{eq:ineq_th3}.
    \end{align} 

Using the bounds in \eqref{eq:max_pr}-\eqref{eq:ineq_th3}, from \eqref{eq:dotv_th3} we have
    \begin{equation}
        \dot V\leq -d_1(t)V+d_2(t)\sqrt{V},
    \end{equation}
    where $d_1(t)=-2\lambda_{\max}(A_\setS^{\mathrm{sym}}(t))$ and
 $d_2(t)=\sqrt{2}\|b_\setS\|+\sqrt{2 M}\bar\mu\|D_\mu f(t)\|+\sqrt{2 M}\bar\nu\|L_\setS(t)\|$.
From H1,
\begin{equation}
    \liminf_{t\to+\infty} d_1(t)\ge 2\mathfrak{h},
\end{equation}
whereas from H2, H3, and H4, and from \eqref{eq:b_term},
\begin{equation}
    \limsup_{t\to+\infty} d_2(t)\le \sqrt{2}\sum_{j=1}^M \bar b_{s_j}+\sqrt{2 M}\bar\mu\mathfrak f+\sqrt{2 M}\bar\nu\mathfrak l.    
\end{equation}

Following similar steps reported in Theorem  \ref{thm:lmi}, we obtain that 
 $$
 \limsup\limits_{t\to+\infty}\|\tilde e_{\mathcal S}(t)\|\le B,
 $$
which implies local boundedness of $e_{\mathcal S}$.
\end{proof}
Notice that Theorem \ref{thm:rob} implies that in the presence of bounded measurement noise and parameter mismatches, the estimation error of the entire network remains bounded, since we can also pick the set $\mathcal S=\mathcal V$.

\begin{figure}
  \centering
  \begin{tabular}{cc}
  \hspace{-5mm}
\includegraphics[scale=0.20]{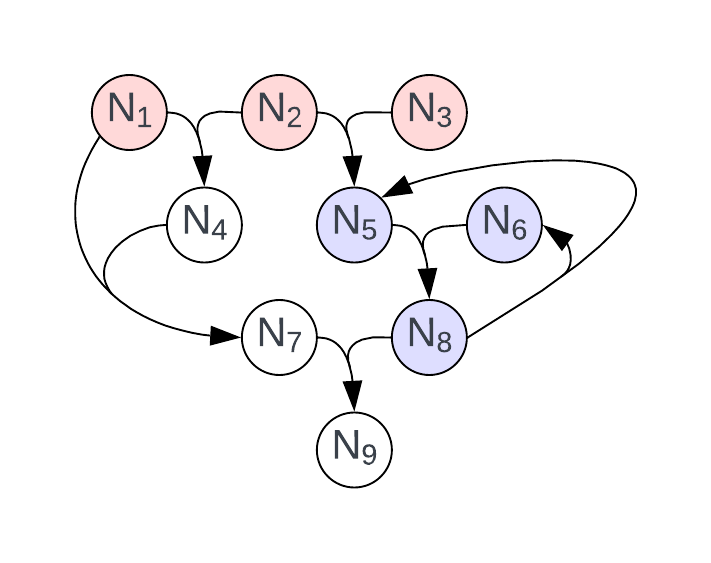}& \hspace{-7mm}\includegraphics[scale=0.20]{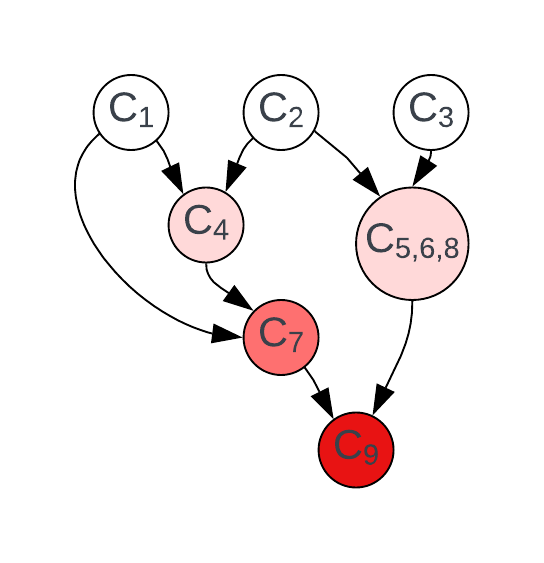}\\
    \hspace{-4mm}$\mathbf a$ & \hspace{-11mm}$\mathbf b$\\
    \hspace{-8mm}\includegraphics[scale=0.50]{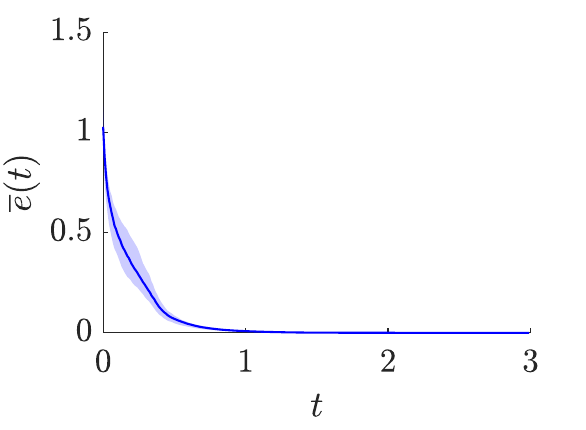} & \hspace{-10mm}\includegraphics[scale=0.50]{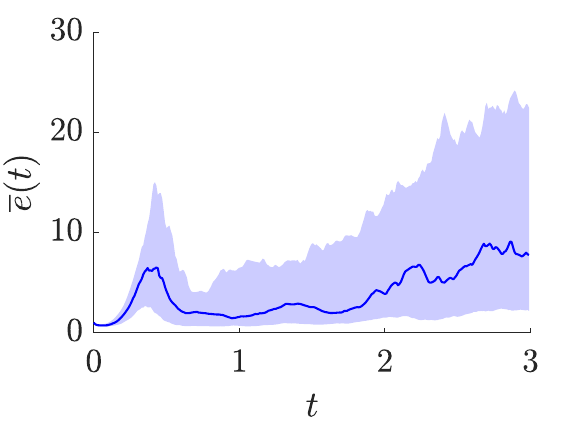} \\
    \hspace{-4mm}$\mathbf c$ & \hspace{-7mm}$\mathbf d$
 
  \end{tabular}
  \caption{Observation of a 9-node network coupled through a directed hypergraph. Panel (a) depicts the hypergraph topology: the nodes that are selected to be measured by Algorithm \ref{alg:main} are colored, with the red ones being the source nodes. Panel (b) reports the condensation of the signed graph associated to the hypergraph in panel (a) to illustrate how Algorithm \ref{alg:main} sequentially builds $\mathcal V_0$: different colors correspond to different iterations of the main loop of Algorithm \ref{alg:main}, with darker shades corresponding the later iterations. Panel (c) and (d) report the median norm $\overline{e}$ of the estimation error, computed over 100 simulations starting from initial conditions for $\hat x_i \sim \mathcal U ([0.8 x_i(0),\,1.2 x_i(0)]) $, $i=1,\ldots,N$, when all the colored nodes of panel (a) and only the sources are measured, respectively; the shaded areas correspond to the interval between the 25th and 75th percentiles of the error norm distributions.}
  \label{fig:toy_topology}
\end{figure}

\section{Numerical validation}\label{sec:validation}

To validate the effectiveness of our results, we start by illustrating the observer design algorithm on a small ($N=9$) network. We consider nodes as coupled Lorenz systems, whose individual dynamics are given by 
\begin{equation}\label{eq:lorenz}
f(x_i)=
\begin{bmatrix}
\mu_{1}(x_{i2}-x_{i1})\\ x_{i1}(\mu_{2}-x_{i3})-x_{i2}\\ x_{i1}x_{i2}-\mu_{3}x_{i3})
\end{bmatrix},
\end{equation}
where $\mu_1=10$, $\mu_2=28$, and $\mu_3=8/3$. The interactions take place through the nonlinear protocol 
\begin{equation}\label{eq:coup_prot_ex}
g(z)=
\begin{bmatrix}
\tilde g(z_{1})&
\tilde g(z_{2})&
\tilde g(z_{3})
\end{bmatrix}^\T,
\end{equation}
where $\tilde g(\zeta)=0.2\zeta+0.05((\zeta+2)\tanh(\zeta+2)-(\zeta-2)\tanh(\zeta-2))$. The coupling strength has been chosen to be homogeneous across the hyperedges, that is, $\sigma_\epsilon=80$ for all $\epsilon\in\mathcal E$, and the network topology $\mathscr H$ is illustrated in Figure \ref{fig:toy_topology}a. Algorithm \ref{alg:main} was used on 100 representative trajectories with initial conditions randomly sampled in the Lorenz attractor. Each trajectory lasted $2$ time units.

A graphical representation of the iterative exploration procedure, performed on the condensation of the signed graph associated to $\mathscr{H}$ is reported in Figure \ref{fig:toy_topology}b. Starting from the root strongly connected components, Algorithm \ref{alg:main} iteratively adds nodes to $\mathcal{V}_o$ by measuring only the output of the nodes that are necessary to fulfill the hypotheses of either Theorem \ref{thm:slowly_varying} or \ref{thm:lmi}.
This algorithmic exploration revealed that measuring the output nodes in  $\mathcal{O}=\{1,2,3,5,6,8\}$ is sufficient to reconstruct the state of the entire network. 

To validate the observer designed using Algorithm \ref{alg:main}, we simulated the network together with its observer 100 times, starting from different initial conditions in $\mathcal{W}=[-3,3]$. Specifically, $x_i(0) \sim \mathcal{U}(\mathcal W)$ and $\hat x_i(0) \sim x_i(0)(1+\mathcal{U}([-0.2,0.2]))$ for all $i\in [1,N]$, with $\mathcal U(z)$ denoting the uniform distribution in the interval $z$.
The norm of the observation error, reported in Fig.\ref{fig:toy_topology}c,  consistently decreases over time and in all the simulations converges to zero with a $5\%$ settling time of at most $0.72$ time units. 
We compare the observation results of Algorithm \ref{alg:main} with a simple heuristic selection of the nodes to be measured, which measures the output only of the RSCCs, which must always be measured to observe the network. However, this condition is not sufficient as illustrated in Figure \ref{fig:toy_topology}d. Here, the observation error never settles and does not approach zero.  

To further validate the proposed algorithm, we next focused on the case of non-invertible output function, and tested the robustness to i) large initial observation error, ii) parametric uncertainties, and iii) measurement errors. Specifically, we considered a larger ($N=20$) group of Lorenz systems coupled through a hierarchical hypergraph (see Appendix \ref{sec:hierarchical_H} for details about the generation of hierarchical hypergraphs). The same coupling protocol $g$ is the same as in \eqref{eq:coup_prot_ex}, $\sigma_\epsilon=\sigma=100$ for all $\epsilon\in\mathcal{E}$. Differently from the smaller scale example above, we select a non-invertible output function, that is, $h(x)=\Gamma x$, with 
\[
\Gamma=\begin{bmatrix}
    0 & 1 & 0\\0 & 0 & 1
\end{bmatrix}.
\]
This implies that the first node state variable is not accessible. Algorithm \ref{alg:main} was used on 100 representative trajectories with initial conditions randomly sampled in the Lorenz attractor. Each trajectory lasted $2$ time units.

The application of our algorithm reveals that by measuring 12 nodes (6 of which are source nodes) we are able to fully reconstruct the network state. This means that the algorithm identifies 6 non-trivial nodes that are instrumental in fully observing the network. In our numerical investigation, we performed 100 simulations starting from $x_0\sim\mathcal U(\mathcal W)$, $\hat x_0\sim x_i(0)(1+\mathcal{U}([-0.2,0.2]))$, and in all instances the observation error norm converges to zero with a $5\%$ settling time of at most  $0.63$ time units, see Figure \ref{fig:ex2}a.

Given that our approach is local, we tested the effectiveness of the algorithm for large uncertainties on the initial conditions of the network, that is, when the initial observation error is large. In particular, we have performed 6 batches of 20 simulations, where in the $\mathfrak b$-th batch the initial conditions for the observer are selected as $\hat x_i(0)\sim x_i(0)(1+\mathcal{U}([-0.5\mathfrak b,0.5 \mathfrak b]))$, $\mathfrak b=1,\ldots,6$. As illustrated in Figure \ref{fig:ex2}b, the solution proposed by Algorithm \ref{alg:main} is robust to large initial observation errors (up to $300\%$), whereby the estimated state asymptotically converges to the network state in all simulations.

In practical applications the actual parameters of the network may differ from the nominal ones. We test the robustness to parametric uncertainties by designing the observer only relying on the nominal value of the parameters, while the agents comprising the network have heterogeneous individual dynamics, with parameters differing from the nominal ones. In Figure \ref{fig:ex2}c, we observed that, in spite of the parametric mismatches, our algorithm is still capable of reconstructing the state with a bounded error, as predicted by our theoretical analysis. Finally, in panel (d) we showed how even in the presence of bounded measurement noise, the network observation error remains bounded, with the maximum error never exceeding 3. 

\begin{figure}
  \centering
  \begin{tabular}{cc}
  \hspace{-5mm}
\includegraphics[scale=0.50]{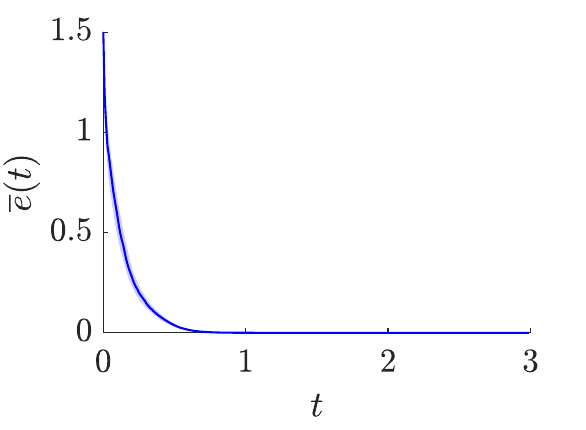}& \hspace{-7mm}\includegraphics[scale=0.50]{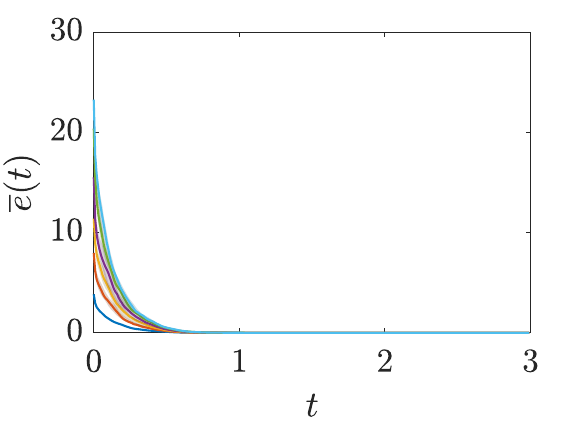}\\
    \hspace{0mm}$\mathbf a$ &\hspace{-4mm} $\mathbf b$\\
    \hspace{-5mm}\includegraphics[scale=0.50]{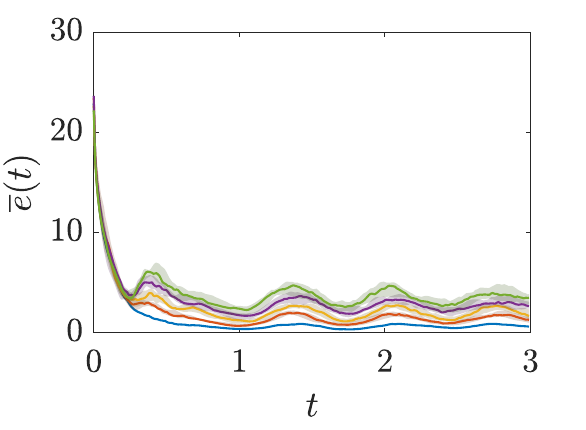} & \hspace{-7mm}\includegraphics[scale=0.50]{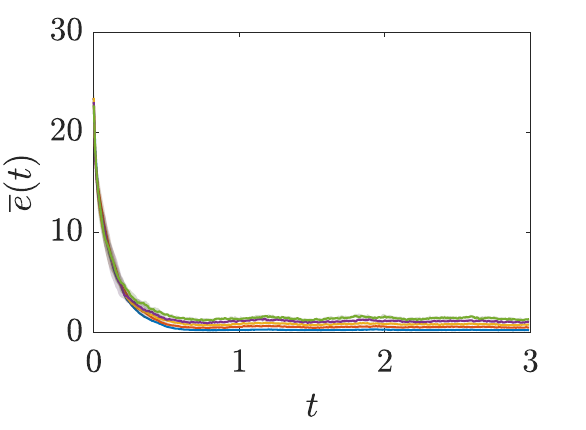} \\
    \hspace{-2mm}$\mathbf c$ & \hspace{-4mm}$\mathbf d$
 
  \end{tabular}
  \caption{Observation of a 20-node network coupled through a hierarchical directed hypergraph. In all panels, the solid lines depict the median norm $\overline{e}$ of the estimation error (computed over 100 simulations in panel (a), and 20 simulations in the other panels, whereas the shaded areas correspond to the interval between the 25th and 75th percentiles of the error norm distributions. Panel (a) considers initial conditions for $\hat x_i \sim \mathcal{U}([0.8 x_i(0),\,1.2 x_i(0)])$, $i=1,\ldots,N$, whereas in panel (b) each of the 6 solid lines (and corresponding shaded areas) correspond to different batches, with the $\mathfrak b$-th batch characterized by initial conditions for the observer $\hat x_i(0)\sim x_i(0)(1+\mathcal{U}([-0.5\mathfrak b,0.5\mathfrak b]))$, $\mathfrak b=1,\ldots,6$, each identified by a different color. Panel (c) depicts the scenario where the actual parameters of the network nodes deviate from the nominal one given in \eqref{eq:lorenz}: each of the 5 solid lines (and corresponding shaded areas) corresponds to a different batch $\mathfrak b$ characterized by $\mu_1\sim 10(1+\mathcal U([-0.02\mathfrak b,+0.02\mathfrak b])$, $\mu_2\sim 28(1+\mathcal U([-0.02\mathfrak b,+0.02\mathfrak b])$, and $\mu_3\sim 8(1+\mathcal U([-0.02\mathfrak b,+0.02\mathfrak b])/3$, $\mathfrak b=1,\ldots,5$, each identified by a different color. Panel (d) depicts the scenario where the output of the $i$-th node is affected by measurement noise $\nu(t)$; each of the 5 solid lines (and corresponding shaded area) corresponds to a different batch $\mathfrak b$ characterized by $\nu(t)\sim\mathcal U([-0.2\mathfrak b, 0.2\mathfrak b])$, $\mathfrak b=1,\ldots,5$, each identified by a different color.}
  \label{fig:ex2}
\end{figure}

\section{Application to opinion dynamics}\label{sec:opinion}

Here, we show how our estimation algorithm can be used to reconstruct the opinions of an ensemble of individuals by only measuring those of a suitably selected subgroup.
Specifically, here we consider the two-option model introduced in \cite{iudice2022bounded}, where the individual dynamics are a bistable system, namely
\[
f(x_i)=-\mu_1x_i+\mu_2\tanh{x_i},
\]
with $\mu_1=1$ and $\mu_2=2$. In the absence of interactions, node $i$ would converge to one of the two stable equilibria, that is, $x^-=-1.91$ and $x^+=1.91$.
 The model captures a binary choice between two alternatives (such as, e.g., an electoral runoff), in which the state of each individual can be interpreted as
its opinion. A positive or negative opinion corresponds to leaning towards one choice or the other. The larger the magnitude of an opinion, the more extreme it will be, since it will be further from 0, which represents the opinion of a neutral, undecided individual. As a coupling protocol, we consider the nonlinear function 
\[
g(z)=0.2z+0.3((z+2)\tanh(z+2)-(z-2)\tanh(z-2))
\]
so that individuals are more sensitive to opinions that are closer to their own.

Social stratification, that is, the hierarchical ranking of people into socioeconomic tiers based on factors like wealth, and education, and power, is a universal phenomenon that divides society into layers \cite{sorokin1927social}.
Such stratification has been observed to also permeate the structure of social networks, leading to hierarchical topological features \cite{jalali2023social}. In order to reproduce such a feature, we have generated a 100-node hierarchical directed hypergraph, and extracted the  largest connected component $\mathscr H_{\mathrm{cc}}$ of the hypergraph, which is composed by 90 nodes and is depicted in Figure \ref{fig:opinion_gerarchico}a. We used $\mathscr H_{\mathrm{cc}}$
to test the effectiveness of our observer in reconstructing the opinions of each individual.

Algorithm \ref{alg:main} was used on 100 representative trajectories with initial conditions sampled, for each node, from the uniform distribution $ \mathcal U([-2.5,2.5])$. Each trajectory lasted $10$ time units. The algorithm identifies 57 nodes whose output needs to be measured, 25 of which are sources. We considered 100 simulations, where we set $\sigma_\epsilon=\sigma=20$ for all $\epsilon \in \mathcal{E}$, which differ for the initial opinion of the individuals, extracted as $x_i(0)\sim \mathcal U([-0.5,0.5])$, and for the initial estimation, obtained as $\hat{x}_i(0)\sim x_i(0)(1+\mathcal U([-0.5,0.5]))$, for $i=1,\ldots,90$.
As illustrated in Figure \ref{fig:opinion_gerarchico}b, the observation error asymptotically vanishes, with a $5\%$ settling time of 3.96 time units. A representative simulation reporting the dynamics of the unmeasured nodes and their corresponding estimations are reported in Figure \ref{fig:opinion_gerarchico}c.

\begin{figure}
  \centering
  \begin{tabular}{c}
  \hspace{-5mm}
\includegraphics[scale=0.70]{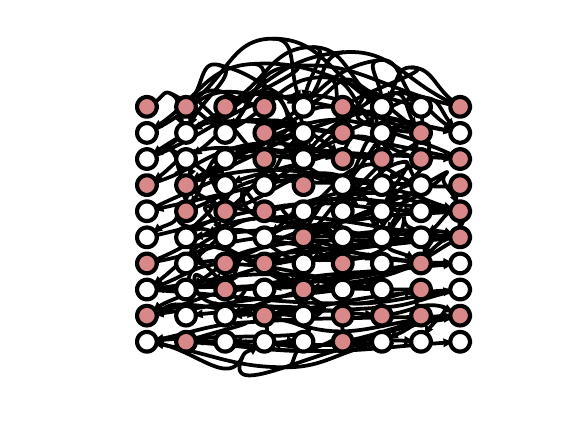}\\
    \hspace{0mm}$\mathbf a$ \\
    \begin{tabular}{cc}
     \hspace{-8mm}\includegraphics[scale=0.50]{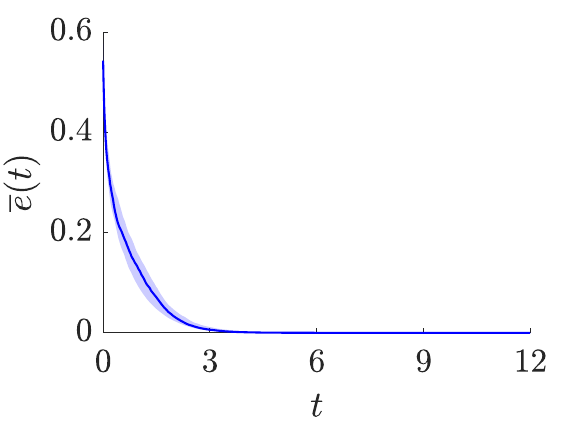} & \hspace{-7mm}\includegraphics[scale=0.50]{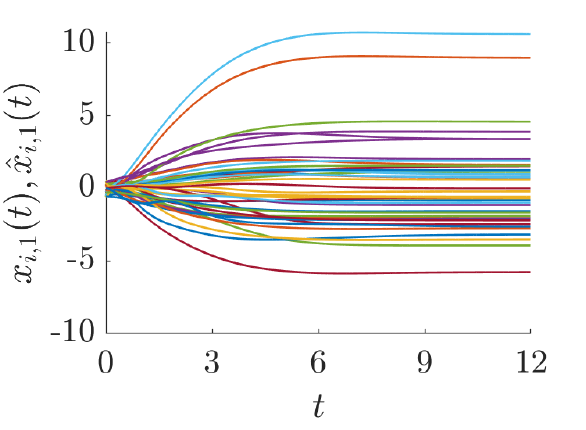} \\
    \hspace{-2mm}$\mathbf b$ & \hspace{-4mm}$\mathbf c$
    \end{tabular}
  \end{tabular}
  \caption{Opinion dynamics over the hierarchical directed hypergraph depicted in panel (a), where the output of the red nodes are measured. Panel (b) reports the median norm $\overline{e}$ of the estimation error, computed over 100 simulations starting from initial conditions for $\hat x_i \sim \mathcal U ([0.5 x_i(0),\,1.5 x_i(0)])$, $i=1,\ldots,N$,  when all the colored nodes of panel (a) are measured; the shaded areas correspond to the interval between the 25th and 75th percentiles of the error norm distributions. Panel (c) depicts the opinions of the nodes that are not measured (solid lines) along with the corresponding estimations (dashed lines) for a representative simulation.}
  \label{fig:opinion_gerarchico}
\end{figure}


\section{Conclusions}\label{sec:conclusion}

In this paper, we tackled the problem of reconstructing the state of a network of nonlinear dynamical systems in the presence of higher-order interactions modeled through directed hypergraphs. We derived analytical conditions that inform the design of an observer guaranteeing asymptotic reconstruction of the network state. Based on these theoretical results, we developed an algorithm that sequentially explores the network, identifying a set of nodes whose output is sufficient to reconstruct the full network state. Through extensive numerical simulations, we have shown that the proposed observer is able to guarantee a bounded reconstruction error even in the presence of parametric uncertainty and measurement error. Furthermore, we have presented a numerical application on opinion dynamics, whereby we have shown how the proposed algorithm is capable of reconstructing the opinion of an entire group of agents by directly accessing the opinion of a suitably selected subset of individuals.

The results presented in this manuscript are promising and pave the way for future investigations. First, the convergence analysis we performed is local; therefore, albeit the numerical analyses suggest that the observer is effective also for large initial estimation error, this is not analytically guaranteed. Alternative approaches may be sought to provide global convergence guarantees. Second, we focus on the relevant case of synchronization noninvasive coupling protocols \cite{gambuzza2021stability,della2023emergence,muolo2025pinning}, with a focus on hyperdiffusive coupling. As in some fields of application, such as ecological networks, the coupling is synchronization invasive, it would be interesting to extend the results to cope with more general coupling protocols.
Finally, we envision that the results presented in this work may also be applied in the context of master-slave synchronization \cite{gutierrez2012targeting}, whereby the observed network can be seen as the master network, and the observer as the slave.

\appendices


\section{Hierarchical directed hypergraphs generation algorithm} \label{sec:hierarchical_H}

Here, we describe the algorithm we used to generate hierarchical directed hypergraphs $\mathscr{H}=(\mathcal V, \mathcal E)$, where all the hyperedges $\epsilon\in\mathcal E$ have the same cardinality $|\epsilon|=\mathfrak c$. The algorithm adds hyperedges with either one tail and $\mathfrak c-1$ heads (denoted source hyperedges) or with one head and $\mathfrak c-1$ tails (denoted sink hyperedges).

The algorithm generates a hierarchical hypergraphs by constructing it in successive layers. Specifically, the nodes of the hypergraphs are partitioned in $\ell$ layers, that is, $\mathcal V=\{\mathcal V_1,\ldots,\mathcal V_\ell\}$. The nodes within the $i$th layer, with $i=1,\ldots,\ell$, are interconnected by $\mathrm{src}_\mathrm{intra}^i$ and $\mathrm{snk}_\mathrm{intra}^i$ source and sink hyperedges, respectively. The heads and tails of each hyperedge $\epsilon$ are uniformly randomly selected within $\mathcal V_i$.

In addition to the hyperedges within each layer, the nodes of the $i$th layer, with $i=1,\ldots,\ell-1$, are also connected with the nodes in the layer $i+1$. Specifically,  $\mathrm{src}_\mathrm{inter}^i$ and $\mathrm{snk}_\mathrm{inter}^i$ source and sink hyperedges are added, whose tails are uniformly randomly selected in $\mathcal V_i$ and whose heads are uniformly randomly selected in $\mathcal V_{i+1}$.

\bibliographystyle{IEEEtran}
\bibliography{IEEEabrv,biblio}

\end{document}